\documentclass[jphysa,amsmath]{iopart}

\usepackage{graphicx}
\usepackage{amsthm}
\usepackage{amsfonts}
\usepackage{amssymb}
\usepackage{braket}
\usepackage{times,txfonts}
\usepackage{subfigure}
\usepackage{hyperref}

\newtheorem{thm}{Theorem}
\newtheorem{prop}{Proposition}
\newtheorem{lemma}{Lemma}

\theoremstyle{definition}

\newcommand{\bq}{\begin{equation*}}
\newcommand{\be}{\begin{equation}}
\newcommand{\eq}{\end{equation*}}
\newcommand{\ee}{\end{equation}}
\newcommand{\bmu}{\begin{multline*}}
\newcommand{\emu}{\end{multline*}}
\newcommand{\ban}{\begin{align*}}
\newcommand{\bal}{\begin{align}}
\newcommand{\ean}{\end{align*}}
\newcommand{\eal}{\end{align}}

\begin{document}
\letter{}

\title{Strong subadditivity for log-determinant of covariance matrices and its applications}

\author{Gerardo Adesso}
\address{{School of Mathematical Sciences, The University of Nottingham, \\
University Park, Nottingham NG7 2RD, United Kingdom}\\ Email: \href{mailto:gerardo.adesso@nottingham.ac.uk}{gerardo.adesso@nottingham.ac.uk}}

\author{R. Simon}
\address{{Optics \& Quantum Information Group,
The Institute of Mathematical Sciences, \\ C.I.T.~Campus, Tharamani, Chennai 600 113, India}\\ Email: \href{mailto:simon@imsc.res.in}{simon@imsc.res.in}}

\begin{abstract}
We prove that the log-determinant of the covariance matrix obeys the strong subadditivity inequality for arbitrary tripartite states of multimode continuous variable quantum systems. This establishes general limitations on the distribution of information encoded in the second moments of canonically conjugate operators. The inequality is shown to be stronger than the conventional strong subadditivity inequality for von Neumann entropy in a class of pure tripartite Gaussian states. We finally show that such an inequality implies a strict monogamy-type constraint for joint Einstein-Podolsky-Rosen steerability of single modes by Gaussian measurements performed on multiple groups of modes.
\end{abstract}

\date{June 13, 2016}

\bigskip

\bigskip 

\section{Introduction}
The formulation of classical information theory,  thanks primarily to the seminal work by Shannon \cite{Shannon}, led to a remarkably broad spectrum of concrete applications in the last century, encompassing in particular systems theory, signal processing, communication and control, complexity and cybernetics. The more recent and still ongoing developments in quantum information theory \cite{Wilde} have opened the way for even more exciting and unprecedented scenarios in the processing of information, with quantum technologies well in the course of revolutionising industrial sectors such as data storage, encryption, sensing, learning, and computing \cite{chuaniels}.

While classical and quantum theory radically differ in the basic set of rules determining the possible and the impossible for the manipulation of information, the two theories rest on some common formal pillars with far-reaching physical implications. Crucial in both cases is in fact the concept of entropy $\mathcal{H}$ as quantifier of information (or, more precisely, of uncertainty), respectively formalised as Shannon entropy $\mathcal{H}_X = - \sum_{i} P(x_i) \log P(x_i)$ for a classical random variable $X$ taking values $\{x_i\}$ with probability distribution $P(x_i)$, and as von Neumann entropy $\mathcal{H}_{\rho} = - {\rm tr} (\rho \log \rho)$ for a quantum state $\rho$.\footnote{Logarithms are usually assumed in base $2$ for finite-dimensional systems and in natural base for infinite-dimensional systems; however, the analysis of this paper does not depend on any specific choice.}

A fundamental limitation for the distribution of entropy in a composite system is then established by the {\it strong subadditivity (SSA) inequality} \cite{ArakiLieb}, which implies the nonnegativity of the mutual information as a measure of total correlations, and  guarantees that the latter quantity never increases upon discarding subsystems.  For a tripartite classical or quantum system $ABC$, this can be formally expressed as
\begin{equation}\label{SSA}
\mathcal{H}_{AB} + \mathcal{H}_{BC} - \mathcal{H}_{A} - \mathcal{H}_{C} \geq 0\,.
\end{equation}
The SSA inequality is straightforward to prove in the classical case, but far less trivial to establish in the quantum case \cite{QuantumSSA,QuantumSSA2}.

 In this Letter we prove that an alternative quantifier of information that can be defined in the quantum case, namely the log-determinant of the covariance matrix of a quantum state, also obeys a SSA inequality formally analogous to Eq.~(\ref{SSA}). We prove that this alternative SSA inequality is stronger than and implies the traditional SSA for the von Neumann entropy, in a class of pure tripartite Gaussian states. We then show that the SSA inequality for log-determinant has important implications for limiting the Einstein-Podolsky-Rosen (EPR) joint steerability \cite{Wiseman} of quantum states in a multipartite setting.

\section{Continuous variable systems and log-determinant of covariance matrices}

  We focus on continuous variable composite quantum systems described by infinite-dimensional Hilbert spaces \cite{Brareview}, as exemplified by a set of $n$ quantum harmonic oscillators (modes).  To describe the most general state $\rho$ of such systems, one requires in principle an infinite hierarchy of moments of the canonically conjugate quadrature operators $\{q_j, p_j\}$ defined on each mode $j=1,\ldots,n$. However, in many practical situations, one can extract already valuable information by considering the first and second moments of the state only. Of these, the first moments play no role in determining any informational quantity, as they can be freely adjusted by local phase space displacements; we shall hence assume vanishing first moments in all the states considered in the following with no loss of generality. What remains central is then the covariance matrix (CM) $V_\rho$, whose elements are defined as \cite{ournewreview}
\begin{equation}\label{Vij}
(V_\rho)_{jk} = {\rm tr} [\rho (R_j R_k + R_k R_j)]
\,,\end{equation}
where $\underline{R}=\{R_1,\ldots, R_{2n}\} = \{q_1,p_1,q_2,p_2,\ldots,q_n,p_n\}$ is the vector of the canonical operators, and we have adopted the natural unit convention such that $V_\rho = \mathbb{I}$ if $\rho$ is the ground (vacuum) state of each oscillator.
Any positive definite, real, symmetric $2n \times 2n$ matrix $V$ is a valid CM of a physical state iff it obeys the {\it bona fide} condition stemming from the uncertainty principle \cite{Simon94}:
\begin{equation}\label{bonafide}
V + i \sigma^{\oplus n} \geq 0\,, \quad \mbox{with $\sigma = \left(
                                                               \begin{array}{cc}
                                                                 0 & 1 \\
                                                                 -1 & 0 \\
                                                               \end{array}
                                                             \right)$}\,.
\end{equation}

The CM of an arbitrary state can be reconstructed efficiently by homodyne detections \cite{Francesi,Virginia}. Confining the description of a state $\rho$ to its CM $V_\rho$ is analogous to implementing a  `small oscillations' approximation for classical oscillators. In the quantum case, for any $\rho$ (with vanishing first moments), one can always define a reference state $\chi_\rho$ uniquely specified by the CM $V_\rho$: the state $\chi_\rho$ will belong to the well-studied class of Gaussian states \cite{pirandolareview,ournewreview},  which are central resources in continuous variable optical and atomic technologies including networked communication, phase estimation and (if supplemented by non-Gaussian detections) one-way quantum computation \cite{Libro,pirandolareview}. Useful sufficient criteria to detect nonclassical correlations such as inseparability \cite{Simon00,Duan00}, steerability \cite{Reid,Wiseman} and nonlocality \cite{Cerf,Nha} of an arbitrary state $\rho$ can be formulated and accessed directly at the level of its CM, although they will typically be necessary only for Gaussian states. In fact, one can quantify the `error' in approximating a state $\rho$ by its CM in terms of the non-Gaussianity of $\rho$, which can be in turn measured by the difference in entropy \cite{Genoni,Wehrl} between the reference Gaussian state $\chi_\rho$ and the original $\rho$.

Quantifying the degree of information (or uncertainty) in a CM $V_\rho$ can thus provide important indications regarding the corresponding properties of any state $\rho$ compatible with such CM, which will be the more accurate the less $\rho$ deviates from Gaussianity. In this respect,  notice that most non-Gaussian resources considered in current protocols (such as photon-subtracted and photon-added states) \cite{Fabio} are constructed as deviations from Gaussian reference states, which means that precious quantitative indications on their degrees of information and (for composite systems) correlations can be gained from the CM alone \cite{Carles,Adesso}. From now on, we shall then speak directly of CMs and measures applied to them.

We define the log-determinant of a CM $V$ as
\begin{equation}\label{logdet}
\mathcal{M}_V = \log (\det V)\,.
\end{equation}
The idea that $\mathcal{M}_V$ may be regarded as an indicator of information akin to (but different from) conventional entropy can be understood as follows. For a Gaussian state $\rho$, the purity is given by ${\rm tr}\,\rho^2=(\det V_\rho)^{-1/2}$, hence  $\mathcal{M}_{V_\rho}$ is a monotonically decreasing function of the purity. Precisely, $\frac12 \mathcal{M}_{V_\rho}$ is equal to the R\'enyi entropy of order $2$ of a Gaussian state with CM $V_\rho$, which is in turn equal to the Shannon entropy of its Wigner quasi-probability distribution (modulo an additive constant) \cite{Renyi}. For a general non-Gaussian state $\rho$ with CM $V_\rho$, we can then interpret $\mathcal{M}_{V_{\rho}}$ as a quantifier of uncertainty in its second moments, expressed by (twice) the R\'enyi entropy of order $2$ of the reference Gaussian state $\chi_\rho$ with the same CM $V_\rho$.

\section{Strong subadditivity for log-determinant and related inequalities}

Given an arbitrary $n$-mode continuous variable system partitioned into three groups of modes forming subsystems $ABC$, with $n_A+n_B+n_C=n$, we denote by $V_\alpha$ and $\mathcal{M}_\alpha$ the CM of (sub)system $\alpha$ and its log-determinant, respectively. In the following, we shall establish the central result of this Letter, announced by the next Theorem.
\begin{thm}[SSA inequality for log-determinant of CMs]\label{T1}
For any tripartite CM $V_{ABC}$, the following inequality holds,
\begin{equation}\label{SSALog}
\mathcal{M}_{AB} + \mathcal{M}_{BC} - \mathcal{M}_{A} - \mathcal{M}_{C} \geq 0\,,
\end{equation}
which by comparison with Eq.~(\ref{SSA}) will be referred to as the  SSA inequality for the log-determinant of the CM.
\end{thm}

Before moving to the proof of the main Theorem, it is instructive to give a simple demonstration of the fact that ordinary (weak) subadditivity holds for log-determinant of CMs.
\begin{prop}[Subadditivity for log-determinant of CMs]\label{P1}
For any bipartite CM $V_{AB}$, the following inequality holds,
\begin{equation}\label{WSALog}
\mathcal{M}_{AB} \leq \mathcal{M}_A + \mathcal{M}_B\,.
 \end{equation}
\end{prop}
\begin{proof}Let $V_{AB}$ be the CM of a bipartite state, with reduced subsystem CMs $V_A$ and $V_B$. In block form, we can write
\begin{equation}\label{blockcm}
V_{AB} = \left(
           \begin{array}{cc}
             V_A & V_{\rm off} \\
             V_{\rm off}^T & V_B \\
           \end{array}
         \right) = L^T            \left(\begin{array}{cc}
             V_A & 0\\
             0 & \bar{V}_{AB\backslash A} \\
           \end{array}
         \right) L\,,
\end{equation}
where \begin{equation}\label{sciuri}
\bar{V}_{AB\backslash A} = V_B - V_{\rm off}^T V_A^{-1} V_{\rm off}
\end{equation} is the Schur complement of $V_A$ in $V_{AB}$, and
$L=\left(
     \begin{array}{cc}
       \mathbb{I} & V_A^{-1} V_{\rm off} \\
       0 & \mathbb{I} \\
     \end{array}
   \right)
$.
Since $V_{\rm off}^T V_A^{-1} V_{\rm off} \geq 0$, we have that $\det \bar{V}_{AB\backslash A} \leq \det V_B$, with equality holding iff $V_{\rm off}=0$. It follows then that $\det V_{AB} \leq \det V_A \det V_B$, which upon taking logarithms implies the claim.
\end{proof}
The inequality (\ref{WSALog}) is analogous to the ordinary subadditivity of entropy, $\mathcal{H}_{AB} \leq \mathcal{H}_A + \mathcal{H}_B\,.$
 A Gaussian state saturates the inequality (\ref{WSALog}) iff it is a product state, but non-Gaussian states can saturate it even if they are not product states, provided their CM takes the direct sum form $V_{AB} = V_A \oplus V_B$. One such example is the non-Gaussian entangled state $\ket{\psi_{AB}}=\ket{00}/\sqrt2 + (\ket{02} +\ket{20})/2$, whose correlations are all in higher order moments \cite{Carles}.

The validity of ordinary subadditivity for log-determinant prompts us to proceed and tackle the proof of the SSA inequality (\ref{SSALog}) announced in Theorem~\ref{T1}. To this end, we make use of two mathematical ingredients.

\begin{lemma}\label{lem1}
The log-determinant is concave over the set of all positive definite matrices.
\end{lemma}
\begin{proof}
This follows from a well known result of classical information theory \cite{logconcave}. Let $V_1, V_2,\ldots, V_l$ be $m \times m$ positive definite matrices; and let $\lambda_1, \lambda_2, \ldots, \lambda_l$ be a set of probabilities, $\lambda_j \geq 0, \sum_j \lambda_j = 1$. Then, $\det \left(\sum_{j=1}^l \lambda_j V_j\right) \geq \prod_{j=1}^l (\det V_j)^{\lambda_j}$. Taking logarithms we obtain the desired concavity property, $
\log \det\left({\sum_j \lambda_j V_j}\right) \geq \sum_j \lambda_j \log\det{V_j}$.
\end{proof}

Lemma~\ref{lem1} establishes that concavity, the primary property of entropy, holds for the log-determinant of CMs. The next auxiliary result we need is as follows.
\begin{lemma}\label{lem2}
The difference $\mathcal{M}_{AB} - \mathcal{M}_{A}$ is concave over the set of all bipartite CMs $V_{AB}$.
\end{lemma}
\begin{proof}
Let $V_{(m)}$ be an $m \times m$ positive definite matrix, and let $V_{(m-l)}$ be the $(m-l) \times (m-l)$ matrix obtained by deleting in $V_{(m)}$ a set of $l$ chosen rows and the corresponding columns. Without loss of generality, $V_{(m-l)}$ can be taken as the leading $(m-l)$-dimensional diagonal block of $V_{(m)}$.  We have then \cite{PetzBook} that $\log \det V_{(m)} - \log\det V_{(m-l)}$ is concave over the set of all $m \times m$ positive definite matrices. Choosing now $V_{(m)}=V_{AB}$ and $V_{(m-l)} = V_A$, the claim is proven.
\end{proof}
Notice that this is true even though the difference in Lemma~\ref{lem2} can be negative, as it does happen for most cases of interest in quantum information theory (e.g.~bipartite entangled states); analogous results hold for the corresponding von Neumann entropic quantity $\mathcal{H}_{AB}-\mathcal{H}_{A}$, whose negativity has been interpreted as a resource for quantum state merging \cite{Nature}.

Equipped with these results, the proof of Theorem~\ref{T1} can now be completed.
\begin{proof}[Proof (of Theorem~\ref{T1})]
Lemma~\ref{lem2} readily implies that \begin{equation}\label{preres}
(\mathcal{M}_{AB} - \mathcal{M}_{A})+(\mathcal{M}_{BC}-\mathcal{M}_C)
\end{equation}
is concave over all tripartite CMs $V_{ABC}$. Since $\{V_{ABC}\}$ form a convex set, concavity implies that the quantity in Eq.~(\ref{preres}) achieves its minimum value at one of the extreme points of this set, i.e., on a CM $V_{ABC}$ with $\det(V_{ABC})=1$. Any such CM describes a pure Gaussian state, for which we have $\det V_{AB} = \det V_C$ and $\det V_{BC} = \det V_A$, which means that the quantity in Eq.~(\ref{preres}) evaluates to zero in any such case. This concludes the proof of the SSA inequality for the log-determinant of CMs as anticipated in Eq.~(\ref{SSALog}).
\end{proof}

We now proceed with some remarks on Theorem~\ref{T1}. First, Proposition~\ref{P1} trivially follows as a corollary of Theorem~\ref{T1}.
Next, we notice that the SSA inequality (\ref{SSALog}) is saturated not only by pure tripartite Gaussian states, but also by all states (Gaussian or not) for which $V_B$ is symplectic (i.e.~such that $V_B \sigma^{\oplus n_B} V_B^T = \sigma^{\oplus n_B}$), that is by states $\rho_{ABC} = \rho_{AC} \otimes \rho_B$ where $\rho_{AC}$ is arbitrary and $\rho_B$ is a pure Gaussian state; the CM of these states can be written in block diagonal form, $V_{ABC} = V_{AC} \oplus V_B$.

If we focus instead on the case of a differently partitioned product state $\rho_{ABC} = \rho_{AB} \otimes \rho_C$ where $\rho_C$ is now a pure Gaussian state, and we consider both the SSA inequality (\ref{SSALog}) and its variant when $A$ and $B$ are swapped, we obtain the inequality \begin{equation}\label{triangle}
\mathcal{M}_{AB} \geq |\mathcal{M}_A - \mathcal{M}_B|\,,
 \end{equation}
 which  is formally analogous to the Araki-Lieb triangle inequality  for Shannon/von Neumann entropies \cite{ArakiLieb,WehrlR}.

Finally, let us recall that, given any CM $V_{ABC}$, it can be `purified' to a symplectic (positive definite) CM $V_{ABCD}$ with $\det V_{ABCD}=1$, so that for global bipartitions of this four-partite CM one has $\det V_{AB} = \det V_{CD}$, $\det V_A = \det V_{BCD}$, and so on. Then, the inequality (\ref{SSALog}) can be recast as (finally relabelling $D$ as $C$ for aesthetic convenience)
\begin{equation}\label{SSALog2}
\mathcal{M}_{AB} + \mathcal{M}_{AC} \geq \mathcal{M}_A + \mathcal{M}_{ABC}\,,
\end{equation}
reminiscent of the celebrated counterpart of Eq.~(\ref{SSA}) for Shannon/von Neumann entropies,
 \begin{equation}\label{SSA2}
\mathcal{H}_{AB} + \mathcal{H}_{AC} \geq \mathcal{H}_A + \mathcal{H}_{ABC}\,,
\end{equation}
which is also typically referred to as SSA inequality in information theory literature. Notice that if $V_{ABC}$ is assumed to be the CM of a tripartite Gaussian state, the inequality (\ref{SSALog2}) reproduces the one demonstrated for the R\'enyi entropy of order $2$ in \cite{Renyi} (see also \cite{Gross}). We remark that in \cite{Renyi} an assumption of Gaussianity of states was made, while here an explicit (and particularly didactic) proof of the SSA for the log-determinant of CMs of arbitrary Gaussian or non-Gaussian states has been presented.

\section{Comparisons between log-determinant and von Neumann entropy}
One might wonder whether there exists a hierarchical relation between the SSA inequalities for the log-determinant $\mathcal{M}_{V_{\rho}}$, Eq.~(\ref{SSALog}), and for the von Neumann entropy $\mathcal{H}_{\rho}$, Eq.~(\ref{SSA}), in arbitrary tripartite states $\rho$. However, the two inequalities are {\it prima facie} incomparable, as it can be seen that they are saturated for different classes of states in general \cite{PetzBook}. Furthermore, while $\mathcal{M}_{V_{\rho}}$ can be computed easily for any state based on second moments, $\mathcal{H}_{\rho}$ does not admit a manageable expression in arbitrary continuous variable states, which renders the comparison even more difficult to undertake.
Nevertheless, if we focus our attention onto Gaussian states, some partial answers can be obtained.

Recall that the von Neumann entropy of an arbitrary $n$-mode Gaussian state  with CM $V$ can be computed in closed form via the expression \cite{Holevo99,Holevo01}
\begin{equation}
\label{vneuG}
{\cal H}_V = \sum_{j=1}^n \frac{\nu_j+1}{2} \log \left(\frac{\nu_j+1}{2}\right) - \frac{\nu_j-1}{2} \log \left(\frac{\nu_j-1}{2}\right)\,,
\end{equation}
where $\{\nu_j\}_{j=1}^n$ are the symplectic eigenvalues of $V$, obeying $\nu_j \geq 1$ $\forall j$ as a consequence of the {\it bona fide} condition (\ref{bonafide}). The latter quantities can be evaluated by noting that the spectrum of the matrix $(-V\sigma^{\oplus n}V\sigma^{\oplus n})$ is of the form  $\{\nu_1^2, \nu_1^2, \ldots, \nu_n^2, \nu_n^2\}$, i.e., contains the squared symplectic eigenvalues of $V$ with double degeneracies.

Exploiting a comparison between various entropies performed in \cite{extremal}, we have that the von Neumann entropy ${\cal H}_V$ admits tight lower and upper bounds as a function of the log-determinant ${\cal M}_V$, for any $n$-mode Guassian state with CM $V$, given by
\begin{equation}\label{HvsM}
f_1({\cal M}_V) \leq {\cal H}_V \leq f_n({\cal M}_V)\,,
\end{equation}
with
\begin{equation}\label{fnm}
f_n(m)=\frac{n}{2} \left[ \log\left(\frac{e^{\frac{m}{n}}-1}{4}\right)
+e^{\frac{m}{2n}} \log\left( \coth \frac{m}{4n}\right) \right]\,.
\end{equation}
For any real $m \geq 0$ and integer $n \geq 1$, the function $f_n(m)$ is monotonically increasing with both $m$ and $n$, and is concave in $m$; furthermore, since $\lim_{m\rightarrow 0^{+}}f_n(m) = 0$, it follows that $f_n(m)$ is also subadditive in $m$. Therefore, for any $x,y \geq 0$,  the following holds, \begin{equation*}f_n(x)+f_n(y) \geq f_n(x+y) \geq \left[f_n(2x)+f_n(2y)\right]/2\,.\end{equation*}

Notice that if we set $n=1$ in Eq.~(\ref{HvsM}), the upper and lower bounds coincide, meaning that the von Neumann entropy is a simple monotonic, concave, and subadditive function of the log-determinant of the CM for all single-mode Gaussian states, while for $n>1$ we can only say that ${\cal H}_V$ is constrained between two monotonic,  concave, and subadditive functions of ${\cal M}_V$, with the upper boundary becoming looser with increasing number $n$ of modes.

We can now show that the SSA inequality for the log-determinant is in fact stronger than the conventional SSA inequality for the von Neumann entropy in a relevant instance.

\begin{thm}[SSA hierarchy for pure Gaussian states] \label{MimpliesH}
 Let $V_{ABC}$ with $\det V_{ABC}=1$ denote the CM of a $(n_A+n_B+n_C)$-mode pure Gaussian state such that the reduced CM $V_A$ has symplectic spectrum $\{1,\ldots,1,\nu_{n_A}\}$.
 Then, the SSA for the log-determinant, Eq.~(\ref{SSALog2}), implies the SSA for the von Neumann entropy, Eq.~(\ref{SSA2}).
\end{thm}

\begin{proof}
For a pure state, ${\cal M}_{ABC}=0$ and Eq.~(\ref{SSALog2}) rewrites as ${\cal M}_A \leq {\cal M}_B +  {\cal M}_C$. If subsystem $A$ is in a Gaussian state whose  CM $V_A$ has $(n_A-1)$ symplectic eigenvalues equal to $1$ (corresponding to $n_A-1$ vacua in its normal mode decomposition), then its entropic properties are equivalent to those of a single mode with symplectic eigenvalue $\nu_{n_A}$, meaning in particular that  the von Neumann entropy of $A$ saturates the lower bound in Eq.~(\ref{HvsM}) \cite{extremal}. We have then the following chain of inequalities:
\begin{equation}
\label{makesyoustronger}
{\cal H}_A = f_1({\cal M}_A) \leq f_1 ({\cal M}_B+{\cal M}_C) \leq f_1({\cal M}_B) + f_1({\cal M}_C) \leq {\cal H}_B + {\cal H}_C\,,
\end{equation}
where we have used respectively the monotonicity of $f_1(m)$ and the SSA for log-determinant in the first inequality, the subadditivity of $f_1(m)$ in the second inequality, and the lower bound of Eq.~(\ref{HvsM}) in the third inequality. Eq.~(\ref{makesyoustronger}) yields ${\cal H}_A \leq {\cal H}_B +  {\cal H}_C$, concluding the proof.
\end{proof}

\begin{figure}[t!]
    \centering
    \includegraphics[width=8cm]{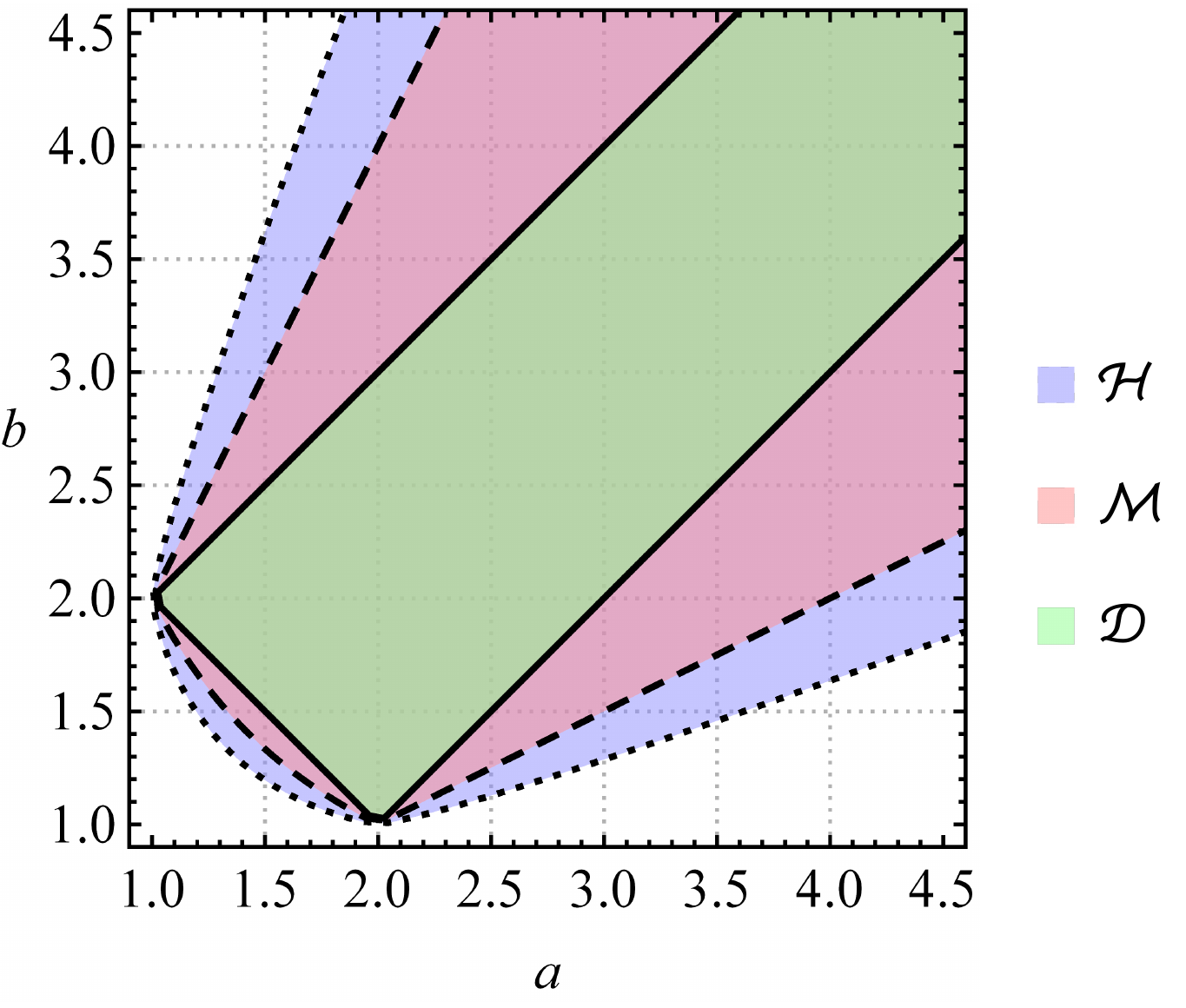}
    \caption{(Colour online) Hierarchy of SSA inequalities as formalised by Theorem~\ref{MimpliesH}. The plot shows a cross-section of the regions defined by the triangle inequality (\ref{gentriangle}) for pure three-mode Gaussian states with local symplectic invariants $\{a,b,c\}$ at fixed $c=2$, with the entropy function ${\cal S}$ corresponding, in order of increasing strength, to:  (i) the von Neumann entropy ${\cal H}$ (blue outermost region with dotted boundary); (ii) the log-determinant of the CM ${\cal M}$ (red intermediate region with dashed boundary); and (iii) the sqrt-determinant of the CM ${\cal D}$ (green innermost region with solid boundary), defined in the text. The latter inequality delimits the physical parameter space of all pure three-mode Gaussian states \cite{3mpra,Giedke2014}.}
    \label{figregia}
\end{figure}

We remark that Theorem~\ref{MimpliesH} holds in particular for all pure tripartite Gaussian states with $n_A=1$ and $n_B, n_C$ arbitrary. In order to provide a simple illustration of the Theorem, let us consider the instance of pure three-mode Gaussian states ($n_A=n_B=n_C=1$). Up to local unitaries, their CM $V_{ABC}$ is fully specified by three symplectic invariants, which can be identified with the determinants of the three reduced CMs, that is, $a=\sqrt{\det V_A}$, $b=\sqrt{\det V_B}$, and $c=\sqrt{\det V_C}$ \cite{3mpra,Giedke2014}. In this case, considering all permutations of the three modes, the SSA constraints take the form of a triangle inequality
\begin{equation}\label{gentriangle}
|{\cal S}_A-{\cal S}_B|\leq {\cal S}_C \leq {\cal S}_A+{\cal S}_B\,,
\end{equation}
with ${\cal S}\equiv {\cal H}$ for the von Neumann entropy, and ${\cal S} \equiv {\cal M}$ for the log-determinant.
In Fig.\ref{figregia} we compare the regions defined by these inequalities in the space of parameters $\{a,b,c\}$. The figure shows, as proven in Theorem~\ref{MimpliesH}, that the log-determinant SSA defines a smaller region and is thus stronger than the von Neumann SSA. However, the set of physical three-mode pure Gaussian states is delimited by an even stronger triangle inequality, obtained by setting ${\cal S} \equiv {\cal D}$ in Eq.~(\ref{gentriangle}), with the sqrt-determinant function ${\cal D}_V=\sqrt{\det V}-1$ \cite{3mpra}. The latter inequality, which can be seen as a solution to the Gaussian marginal problem for $n=3$ \cite{Tyc2008}, further incorporates the requirement that the CM $V_{ABC}$ must obey the {\it bona fide} condition (\ref{bonafide}), while the SSA inequality in the form (\ref{SSALog2}) for the log-determinant only relies on the positivity of the CM, $V_{ABC}>0$ (see also \cite{Renyi,Gross,PetzBook}),  which is weaker than Eq.~(\ref{bonafide}).

\section{Applications to EPR steering of multimode states}

In the remaining part of the Letter, we investigate applications of the SSA inequality (\ref{SSALog}) for log-determinant of CMs to characterising possibilities and limitations of EPR steering in continuous variable systems. Let us briefly introduce the necessary concepts. Steering, intended in a bipartite setting as the possibility for Alice to remotely prepare Bob's system in different states depending on her own local measurements,  is a genuine manifestation of quantum correlations that embodies the crux of the original EPR paradox \cite{EPR}, and was recognised by Schr\"odinger as evidence of the ``amazing knowledge'' allowed by quantum mechanics \cite{Schr,Schr2}.

Let  $\rho_{AB}$ be a bipartite state, and let  $a$ and $b$ be measurement operators on subsystems $A$ (operated by Alice) and $B$ (operated by Bob), with respective outcomes $\alpha$ and $\beta$. By definition \cite{Wiseman}, the state $\rho_{AB}$ is $A \to B$ steerable iff, for all pairs $a$ and $b$, the measurement statistics obeys
\begin{equation}\label{steeringdefinition}
\mbox{$p(\alpha, \beta|a,b;\rho_{AB}) \neq \sum_\lambda p_\lambda p(\alpha|a,\lambda) p(\beta|b,\tau_B^\lambda)$}\,,
\end{equation}
that is, it cannot be interpreted as arising from correlations between a random local hidden variable ($\lambda$) for Alice and a random local hidden state ($\tau_B^\lambda$) measured by Bob. Here $p_\lambda$ is a probability distribution and $p(\alpha, \beta|a,b;\rho_{AB}) = {\rm tr} [(\Pi_A^a \otimes \Pi_B^b)\rho_{AB}]$,
where $\Pi_A^{a}$ is the projector satisfying $a\Pi_A^a = \alpha \Pi_A^a$.

For a two-mode continuous variable system, a sufficient condition to detect steerability \cite{Reid} can be expressed in terms of the violation of Heisenberg-type uncertainty relations for the conditional variances corresponding to measurements of canonically conjugate operators. Let $q_B$ and $p_B$ be quadrature operators on Bob's mode, satisfying $[q_B,p_B]=i$, and define the variances that Alice deduces for Bob by linear inference based on her own measurements of a pair of uncharacterised operators $q_A$ and $p_A$, e.g.~$\Delta_{{\rm inf},A}q_B = \langle  [q_B-q_B^{\rm est}(q_A)]^2\rangle_{\rho_{AB}}$, where $q_B^{\rm est}(q_A) = g_q q_A$ for some optimised value of the linear gain coefficient $g_q$ (and similarly for $p_B$). One has then that the state $\rho_{AB}$ is $A \rightarrow B$ steerable if \cite{Reid}
\begin{equation}
E_{B|A}(\rho_{AB})= \Delta_{{\rm inf},A}q_B\  \Delta_{{\rm inf},A} p_B <1 \label{reid}
\end{equation}
The criterion in Eq.~(\ref{reid}) can detect steerability due to quadrature (Gaussian) measurements which act on second moments. If one optimises it over all possible choices of canonically conjugate pairs (i.e.~over local phase space symplectic transformations for Alice and Bob), then the minimum of $E_{A|B}$ can be expressed only in terms of the CM $V_{AB}$ of $\rho_{AB}$ \cite{Kogias,WisePRA,JOSAB}:
\begin{equation}
\min_{U_A \otimes U_B} E_{B|A}\big((U_A \otimes U_B)\rho_{AB}(U_A \otimes U_B)^{\dagger}\big) = \det \bar{V}_{AB \backslash A} = \frac{\det V_{AB}}{\det V_A}\,,
 \end{equation}
 where $\bar{V}_{AB \backslash A}$ denotes the Schur complement of $V_{A}$ in $V_{AB}$ as in Eq.~(\ref{sciuri}).  In this form, the criterion can be extended to an arbitrary number of modes: given a bipartite state $\rho_{AB}$ with CM $V_{AB}$, if the Schur complement is not itself a bona fide CM in the sense of Eq.~(\ref{bonafide}), i.e.~if
\begin{equation}\label{bonasteer}
\bar{V}_{AB \backslash A} + i \sigma^{\oplus n_B} \not\geq 0\,,
\end{equation}
then $\rho_{AB}$ is $A \rightarrow B$ steerable. Steerable states are useful resources for one-sided device-independent quantum key distribution \cite{1SQKD}, subchannel discrimination \cite{Piani}, and secure continuous variable teleportation \cite{TeleSteer}.

In the special case of bipartite Gaussian states $\rho_{AB}$, Eq.~(\ref{bonasteer}) is necessary and sufficient for steerability by Gaussian measurements \cite{Wiseman,WisePRA}. Accordingly, a quantitative measure of Gaussian steerability has been proposed for a $(n_A+n_B)$-mode bipartite Gaussian state \cite{Kogias}, defined as
\begin{equation}\label{GSAtoB}
{\cal G}^{A \to B}(V_{AB})=
\left\{
  \begin{array}{ll}
    0, &\hspace*{-2.2cm} \hbox{$\bar{\nu}^{AB\backslash A}_j \geq 1$ $\forall j=1,\ldots,n_B$\ ;} \\
    -\sum_{j:\bar{\nu}^{AB\backslash A}_j<1} \log\left(\bar{\nu}^{AB\backslash A}_j\right), & \hbox{otherwise,}
  \end{array}
\right.
\end{equation}
where $\{\bar{\nu}^{AB\backslash A}_j\}$ denote the symplectic eigenvalues of $\bar{V}_{AB \backslash A}$.
In the special case of $B$ comprising one mode only ($n_B=1$), the Schur complement $\bar{V}_{AB \backslash A}$ has only one symplectic eigenvalue $\bar{\nu}^{AB\backslash A}=\sqrt{\det{\bar{V}_{AB \backslash A}}}$, hence the above expression simplifies to
\begin{equation}\label{GSAtoB1}
\left.{\cal G}^{A \to B}(V_{AB})\right\vert_{n_B=1} = \max\left\{0,\, \frac12 (\mathcal{M}_A - \mathcal{M}_{AB})\right\}\,,
\end{equation}
where we have adopted the expression in Eq.~(\ref{logdet}) for the log-determinant. The log-determinant is therefore useful to capture the quantitative degree of steerability of mode $B$ by Gaussian measurements performed on the multimode subsystem $A$, as detectable at the level of CMs. Notice however that, very recently, examples of Gaussian states unsteerable by Gaussian measurements (i.e.~with  ${\cal G}^{A \to B}=0$) have been found, which are nonetheless steerable, according to definition (\ref{steeringdefinition}), by means of suitable non-Gaussian measurements \cite{NoGauss1,NoGauss2}.

Consider now a tripartite setting. Quite interestingly, if only quadrature (Gaussian) measurements and second moments are considered for steering detection, then a very strong limitation occurs: for an arbitrary (Gaussian or non-Gaussian) three-mode state $\rho_{ABC}$, the {\it monogamy} constraint
\begin{equation}\label{ReidMono}
E_{B|A}(\rho_{ABC})\  E_{B|C}(\rho_{ABC}) \geq 1\,,
\end{equation}
holds \cite{ReidMono,HeReid,Armstrong}. This means that it is impossible to detect any simultaneous steering of the single mode $B$ by the single modes $A$ and $C$ if using second moment criteria.

\begin{figure}[t!]
    \centering
    \includegraphics[width=8cm]{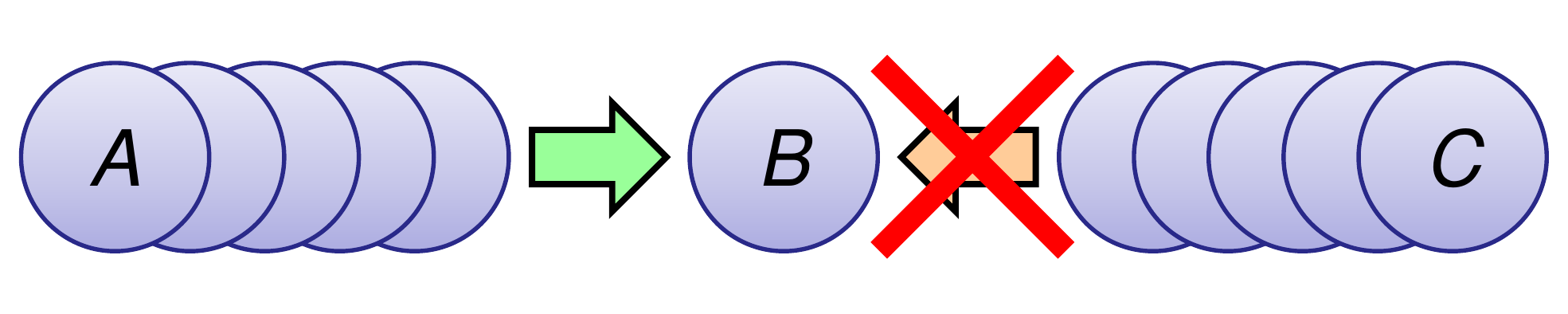}
    \caption{(Colour online) Monogamy of EPR steering by Gaussian measurements as formalised by Theorem~\ref{theosteer}. For any state $\rho_{ABC}$ of a multimode system partitioned into three subsystems $ABC$, where $A$ and $C$ are composed of an arbitrary number of modes while $B$ is composed of a single mode, if $A$ can steer $B$ by Gaussian measurements, then $C$ cannot steer $B$ by Gaussian measurements, and vice versa. This is a direct consequence of the SSA inequality for the log-determinant of the CM of $\rho_{ABC}$ presented in Theorem~\ref{T1}.}
    \label{fig2}
\end{figure}

We can now generalise this result to the case of parties $A$ and $C$ comprising an arbitrary number of modes, while the steered party $B$ remains formed by a single mode (see Fig.~\ref{fig2}).

\begin{thm}[No-joint steerability of one mode by multimode Gaussian measurements]\label{theosteer}
Let $\rho_{ABC}$ be an arbitrary $(n_A+n_B+n_C)$-mode quantum state with $n_A, n_C$ arbitrary and $n_B=1$. Then it is impossible for $\rho_{ABC}$ to be simultaneously $A \rightarrow B$ and $C \rightarrow B$ steerable by Gaussian measurements.
\end{thm}

\begin{proof}
The claim follows by combining the SSA inequality (\ref{SSALog}) for log-determinant of CMs with the CM-based steering criterion (\ref{bonasteer}). Denoting by $V_{ABC}$ the CM of the composite tripartite system (with $n_B=1$), we have in fact that $\bar{V}_{AB \backslash A} + i \sigma^{\oplus n_B} \not\geq 0$ is equivalent to $\det \bar{V}_{AB \backslash A} = \det V_{AB}/\det V_A<1$, and similarly $\bar{V}_{CB \backslash C} + i \sigma^{\oplus n_B} \not\geq 0$ is equivalent to $\det \bar{V}_{CB \backslash C} = \det V_{BC}/\det V_C<1$. To accomplish simultaneous steering of mode $B$ by groups $A$ and $C$ based on second moments, one would thus need $(\det V_{AB}/\det V_A)(\det V_{BC}/\det V_C) < 1$, or equivalently, taking logarithms, $\mathcal{M}_{AB} + \mathcal{M}_{BC} - \mathcal{M}_{A} - \mathcal{M}_{C} <1$. But this is impossible as it contradicts Eq.~(\ref{SSALog}), hence concluding the proof.
\end{proof}
The SSA inequality for the log-determinant implies therefore a limitation for joint steerability based solely on CMs in arbitrary continuous variable states.  Specifically, subsystems $A$ and $C$ cannot simultaneously steer the single-mode subsystem $B$ by Gaussian measurements, although this no-go may be circumvented by non-Gaussian measurements even on tripartite Gaussian states $\rho_{ABC}$ \cite{NoGauss2}.

However, it is easy to see that, as soon as $B$ is made of at least two modes, such a strict monogamy is lifted, and $B$ can be steered by parties $A$ and $C$ simultaneously already in an all-Gaussian setting. As an example, consider a Gaussian state of four modes $1,2,3,4$, and group them such that subsystem $A$ is assigned mode $1$, subsystem $B$ is assigned modes $2$ and $3$, and subsystem $C$ is assigned mode $4$. For illustration, we can focus on the family of four-mode pure states introduced in \cite{Promis} (see Figure 1 therein), whose CM takes the form
\begin{equation}\label{s4}
V_{ABC} \equiv V_{1234} =S_{3,4}(a)S_{1,2}(a)S_{2,3}(s)S_{2,3}^T (s)S_{1,2}^T
(a)S_{3,4}^T (a)\,,
\end{equation}
where $S_{i,j}(r)$ denotes a two-mode squeezing symplectic transformation acting on modes $i$ and $j$ with real squeezing degree $r$ \cite{ournewreview}. These states are symmetric under swapping of modes $1 \leftrightarrow 4$ and $2 \leftrightarrow 3$. Therefore, their $1 \rightarrow (2,3)$ and $4 \rightarrow (2,3)$ steerability properties are the same; with respect to our grouping, modes $A$ and $C$ are thus able in principle to steer simultaneously the two-mode group $B$ by the same amount.  To see whether any such steering is possible at all, we can calculate the Gaussian steerability measure (\ref{GSAtoB}) \cite{Kogias} in the relevant settings. We find that as soon as $a,s>0$, i.e., as soon as the state is not a product state, then it is both $A \rightarrow B$ and $C \rightarrow B$ steerable by Gaussian measurements,  with its Gaussian steerability ${\cal G}^{A \to B}(V_{ABC})= {\cal G}^{C \to B}(V_{ABC})$ being a monotonically increasing function of $a$ and $s$, not reported here. This does not contradict the general SSA inequality (\ref{SSALog}), which holds with equality on this example as $\det V_{ABC}=1$.
The explanation  is that, when the steered party has more than one mode, the symplectic spectrum entering Eq.~(\ref{GSAtoB}) does not depend only on the determinant of the Schur complement, hence steerability cannot be decided solely in terms of a balance of log-determinants.

\section{Conclusions}

In this Letter we demonstrated that the log-determinant, a simple informational quantity defined on the covariance matrix of any continuous variable state, behaves as a fully fledged  entropy, obeying the fundamental strong subadditivity inequality.  In a particular class of pure tripartite Gaussian states of an arbitrary number of modes, we showed that such a constraint is stronger than the conventional strong subadditivity inequality for von Neumann entropy. It would be very interesting as a future direction to investigate whether this hierarchy between strong subadditivities holds true in general, or may be reversed on other classes of states.

Our result implies a strict limitation on the joint steerability of one quantum harmonic oscillator by two other groups of oscillators, within a steering detection setting based on second moments. This is in turn relevant for practical applications, e.g. in the context
of secure quantum communication \cite{pirandolareview}. In a typical quantum optics laboratory where operations (including malicious attacks) are limited to the Gaussian toolbox, it is impossible for a single mode in Bob's possession to be steered by more than one partner at once. Such a monogamy ensures that Bob's exclusive pairing with Alice (who can operate on multiple modes), for the purposes of entanglement verification \cite{Wiseman} and one-sided device-independent quantum key distribution \cite{1SQKD}, cannot be disrupted by the attempts of an eavesdropper Claire. It will be interesting to investigate other applications of this no-go result in the context of secure teleportation and telecloning protocols involving three or more parties \cite{TeleSteer}.

More extensions and additional strenghtenings of the strong subadditivity inequality for log-determinant of covariance matrices, taking into account physical requirements  such as the uncertainty principle, and inspired by seminal or  more modern developments in classical and quantum information theory, are certainly worthy of further investigation \cite{Gross,Lami}. In particular, we anticipate that it is possible to define a remainder term for Eq.~(\ref{SSALog2}) by means of a Gaussian recovery map \cite{Lami}, in analogy to the latest advances obtained for the strong subadditivity of von Neumann entropy by Fawzi and Renner \cite{fawzi2015quantum}. A more comprehensive characterisation of (Gaussian and non-Gaussian) states saturating the strong subadditivity inequality for log-determinat also deserves a separate study. We finally notice that other monogamy-type constraints on continuous variable steering within multipartite networks have been recently explored and reported elsewhere \cite{Yu}.

\section*{Acknowledgments}
GA is grateful to the Institute of Mathematical Sciences (Chennai, India) for the kind hospitality during completion of an early draft of this work, and acknowledges very fruitful discussions with Ioannis Kogias, Antony Lee, Yu Xiang, Qiongyi He, David Gross, Michael Walter, Andreas Winter, Christoph Hirche, and especially Ludovico Lami. GA  is supported by the European Research Council (ERC StG GQCOP, Grant No.~637352). RS is grateful to the Science and Engineering Research Board, Government of India for a {\em SERB Distinguished Fellowship} which made this work possible.

\section*{References}
	

\begin{thebibliography}{10}
\expandafter\ifx\csname url\endcsname\relax
  \def\url#1{{\tt #1}}\fi
\expandafter\ifx\csname urlprefix\endcsname\relax\def\urlprefix{URL }\fi
\providecommand{\eprint}[2][]{\url{#2}}

\bibitem{Shannon}
Shannon C~E 1948 {\em The Bell System Technical Journal\/} {\bf 27} 379--423

\bibitem{Wilde}
Wilde M~M 2013 {\em Quantum Information Theory\/} (Cambridge University Press,
  Cambridge)

\bibitem{chuaniels}
Nielsen M~A and Chuang I~L 2000 {\em Quantum Computation and Quantum
  Information\/} (Cambridge: Cambridge University Press, Cambridge)

\bibitem{ArakiLieb}
Araki H and Lieb H 1970 {\em Commun. Math. Phys.\/} {\bf 18} 160

\bibitem{QuantumSSA}
Lieb E~H and Ruskai M~B 1973 {\em J. Math. Phys.\/} {\bf 14} 1938

\bibitem{QuantumSSA2}
Nielsen M~A and Petz D 2005 {\em Quant. Inf. Comput.\/} {\bf 5} 507

\bibitem{Wiseman}
Wiseman H~M, Jones S~J and Doherty A~C 2007 {\em Phys. Rev. Lett.\/} {\bf
  98}(14) 140402

\bibitem{Brareview}
Braunstein S~L and van Loock P 2005 {\em Rev. Mod. Phys.\/} {\bf 77} 513

\bibitem{ournewreview}
Adesso G, Ragy S and Lee A~R 2014 {\em Open Syst. Inf. Dyn.\/} {\bf 21} 1440001

\bibitem{Simon94}
Simon R, Mukunda N and Dutta B 1994 {\em Phys. Rev. A\/} {\bf 49}(3) 1567--1583

\bibitem{Francesi}
Laurat J, Keller G, {Oliveira-Huguenin} J~A, Fabre C, Coudreau T, Serafini A,
  Adesso G and Illuminati F 2005 {\em J. Opt. B: Quant. Semiclass. Opt.\/} {\bf
  7} S577

\bibitem{Virginia}
D'Auria V, Fornaro S, Porzio A, Solimeno S, Olivares S and Paris M~G~A 2009
  {\em Phys. Rev. Lett.\/} {\bf 102}(2) 020502

\bibitem{pirandolareview}
Weedbrook C, Pirandola S, Garcia-Patron R, Cerf N~J, Ralph T~C, Shapiro J~H and
  Lloyd S 2012 {\em Rev. Mod. Phys.\/} {\bf 84} 621

\bibitem{Libro}
Cerf N~J, Leuchs G and Polzik E~S 2007 {\em {Quantum Information with
  Continuous Variables of Atoms and Light}\/} (Imperial College Press) ISBN
  1860947603

\bibitem{Simon00}
Simon R 2000 {\em Phys. Rev. Lett.\/} {\bf 84} 2726

\bibitem{Duan00}
Duan L~M, Giedke G, Cirac J~I and Zoller P 2000 {\em Phys. Rev. Lett.\/} {\bf
  84} 2722

\bibitem{Reid}
Reid M~D 1989 {\em Phys. Rev. A\/} {\bf 40}(2) 913--923

\bibitem{Cerf}
Garc\'{\i}a-Patr\'on R, Fiur\'a\ifmmode~\check{s}\else \v{s}\fi{}ek J, Cerf
  N~J, Wenger J, Tualle-Brouri R and Grangier P 2004 {\em Phys. Rev. Lett.\/}
  {\bf 93}(13) 130409

\bibitem{Nha}
Nha H and Carmichael H~J 2004 {\em Phys. Rev. Lett.\/} {\bf 93}(2) 020401

\bibitem{Genoni}
Genoni M~G, Paris M~G~A and Banaszek K 2008 {\em Phys. Rev. A\/} {\bf 78}(6)
  060303(R)

\bibitem{Wehrl}
Ivan J~S, Kumar M~S and Simon R 2012 {\em Quant. Inf. Proc.\/} {\bf 11}
  853--872 ISSN 1570-0755

\bibitem{Fabio}
Dell'Anno F, Siena S~D and Illuminati F 2006 {\em Phys. Rep.\/} {\bf 428} 53

\bibitem{Carles}
Rod\'o C, Adesso G and Sanpera A 2008 {\em Phys. Rev. Lett.\/} {\bf 100}(11)
  110505

\bibitem{Adesso}
Adesso G 2009 {\em Phys. Rev. A\/} {\bf 79}(2) 022315

\bibitem{Renyi}
Adesso G, Girolami D and Serafini A 2012 {\em Phys. Rev. Lett.\/} {\bf 109}
  190502

\bibitem{logconcave}
Cover T~M and Thomas A 1988 {\em SIAM Journal on Matrix Analysis and
  Applications\/} {\bf 9} 384

\bibitem{PetzBook}
Hiai F and Petz D 2014 {\em Introduction to Matrix Analysis and Applications\/}
  Universitext (Springer) ISBN T978-3319041490

\bibitem{Nature}
Horodecki M, Oppenheim J and Winter A 2005 {\em Nature\/} {\bf 436} 673--676
  ISSN 0028-0836

\bibitem{WehrlR}
Wehrl A 1978 {\em Rev. Mod. Phys.\/} {\bf 50}(2) 221--260

\bibitem{Gross}
Gross D and Walter M 2013 {\em J. Math. Phys.\/} {\bf 54} 082201

\bibitem{Holevo99}
Holevo A~S, Sohma M and Hirota O 1999 {\em Phys. Rev. A\/} {\bf 59} 1820

\bibitem{Holevo01}
Holevo A~S and Werner R~F 2001 {\em Phys. Rev. A\/} {\bf 63} 032312

\bibitem{extremal}
Adesso G, Serafini A and Illuminati F 2004 {\em Phys. Rev. A\/} {\bf 70} 022318

\bibitem{3mpra}
Adesso G, Serafini A and Illuminati F 2006 {\em Phys. Rev. A\/} {\bf 73} 032345

\bibitem{Giedke2014}
Giedke G and Kraus B 2014 {\em Phys. Rev. A\/} {\bf 89}(1) 012335

\bibitem{Tyc2008}
Eisert J, Tyc T, Rudolph T and Sanders B~C 2008 {\em Commun. Math. Phys.\/}
  {\bf 280} 263 ISSN 1432-0916

\bibitem{EPR}
Einstein A, Podolsky B and Rosen N 1935 {\em Phys. Rev.\/} {\bf 47}(10)
  777--780

\bibitem{Schr}
Schr\"odinger E 1935 {\em Proc. Camb. Phil. Soc.\/} {\bf 31} 553

\bibitem{Schr2}
Schr\"odinger E 1936 {\em Proc. Camb. Phil. Soc.\/} {\bf 32} 446

\bibitem{Kogias}
Kogias I, Lee A~R, Ragy S and Adesso G 2015 {\em Phys. Rev. Lett.\/} {\bf
  114}(6) 060403

\bibitem{WisePRA}
Jones S~J, Wiseman H~M and Doherty A~C 2007 {\em Phys. Rev. A\/} {\bf 76}(5)
  052116

\bibitem{JOSAB}
Kogias I and Adesso G 2015 {\em J. Opt. Soc. Am. B\/} {\bf 32} A27

\bibitem{1SQKD}
Walk N, Wiseman H~M and Ralph T~C 2014 {\em arXiv:1405.6593 [quant-ph]\/}

\bibitem{Piani}
Piani M and Watrous J 2015 {\em Phys. Rev. Lett.\/} {\bf 114}(6) 060404

\bibitem{TeleSteer}
He Q, Rosales-Z\'arate L, Adesso G and Reid M~D 2015 {\em Phys. Rev. Lett.\/}
  {\bf 115}(18) 180502

\bibitem{NoGauss1}
Wollmann S, Walk N, Bennet A~J, Wiseman H~M and Pryde G~J 2016 {\em Phys. Rev.
  Lett.\/} {\bf 116}(16) 160403

\bibitem{NoGauss2}
Ji S~W, Lee J, Park J and Nha H 2015 {\em arXiv:1511.02649 [quant-ph]\/}

\bibitem{ReidMono}
Reid M~D 2013 {\em Phys. Rev. A\/} {\bf 88}(6) 062108

\bibitem{HeReid}
He Q~Y and Reid M~D 2013 {\em Phys. Rev. Lett.\/} {\bf 111}(25) 250403

\bibitem{Armstrong}
Armstrong S, Wang M, Teh R~Y, Gong Q, He Q, Janousek J, Bachor H~A, Reid M~D
  and Lam P~K 2015 {\em Nat. Phys.\/} {\bf 11} 167

\bibitem{Promis}
Adesso G, Ericsson M and Illuminati F 2007 {\em Phys. Rev. A\/} {\bf 76}(2)
  022315

\bibitem{Lami}
Lami L {\it et al} 2016 in preparation

\bibitem{fawzi2015quantum}
Fawzi O and Renner R 2015 {\em Commun. Math. Phys.\/} {\bf 340} 575--611

\bibitem{Yu}
Xiang Y, Kogias I, Adesso G and He Q 2016 {\em arXiv:1603.08173 [quant-ph]\/}

\end{thebibliography}

\providecommand{\newblock}{}

\end{document}